\documentclass{amsart}
\RequirePackage{fix-cm}
\usepackage{graphicx}

\usepackage{amssymb,amsmath,amsthm,enumerate,graphicx,mathtools,subcaption, url}
\usepackage{stmaryrd}
\usepackage{amssymb}

\usepackage{amsbsy,amsmath}
\usepackage{amsfonts}
\usepackage{graphicx}
\usepackage{ulem}
\usepackage{breakcites}

\usepackage{breakcites}
\usepackage[margin=1.3in]{geometry}
\usepackage{multirow}

\newcommand{\norm}[1]{\left\lVert#1\right\rVert}
\newcommand{\DOF}{DoF \,}

\newcommand{\ie}{\textit{i.e.,\,}}
\newcommand{\eg}{\textit{e.g.,\,}}

\newtheorem{prop}{Proposition} 
\newtheorem{reeemark}{Remark}

\usepackage{rotating}

\usepackage{hyperref}
\begin{document}

\title[Detection of separatrices and chaotic seas based on orbit amplitudes]{Detection of separatrices and chaotic seas based on orbit amplitudes
}

\author[J.\,Daquin]{
J\'er\^ome Daquin         
}
\address{Department of Mathematics (naXys), $61$ Avenue de Bruxelles, $5000$, Namur, Belgium}
\email{jerome.daquin@unamur.be}

\author[C.\,Charalambous]{Carolina Charalambous}
\address{Department of Mathematics (naXys), $61$ Avenue de Bruxelles, $5000$, Namur, Belgium}
\email{carolina.charalambous@unamur.be}

\address{}

\date{\today}

\date{\today}

\maketitle

\begin{abstract}
The Maximum Eccentricity Method (MEM, \cite{rDv04}) is a standard tool for the analysis of planetary  systems and their stability. 
The method amounts to estimating the maximal stretch of orbits over sampled domains of initial conditions. The present paper leverages on the MEM to introduce a sharp detector of separatrices and chaotic seas. After introducing the MEM analogue for nearly-integrable action-angle Hamiltonians, \ie 
diameters, we  use low-dimensional dynamical systems with multi-resonant modes and junctions, supporting chaotic motions, to recognise the drivers of the diameter metric.
Once this is appreciated, we present a second-derivative based index measuring the regularity of this application. This quantity turns to be a sensitive and robust indicator to detect separatrices, resonant webs and chaotic seas. We  discuss practical applications of this framework in the context of $N$-body simulations for the planetary case affected by mean-motion resonances, and demonstrate the ability of the index to distinguish minute structures of the phase space, otherwise undetected with the original MEM. 
\end{abstract}	

\tableofcontents

\section{Introduction}\label{sec:intro}
Among the 5,000 discovered extrasolar systems\footnote{Confer The Extrasolar Planets Encyclopaedia, \href{http://exoplanet.eu/catalog}{http://exoplanet.eu} \cite{jSc11}.}, 
there is a great diversity in multiplicities of planets and central stars, orbital architectures (spacing between planets, orbital parameters), or physical properties (masses, radii). 
Being able to understand the complex gravitational interactions of this kaleidoscope, such as the dynamics associated to resonant configurations, is a dynamical challenge that
help to retrace their histories since formation, and further constrain observational campaigns (see, \eg \cite{cGi12,sHa20,aCe21,jTe22,mSt22}).  
The Maximum Eccentricity Method (hereafter, MEM) introduced by \cite{rDv04} originated to probe the dynamical structure and largeness of stability regions of exoplanetary systems in the context of direct numerical integrations of the equations of motion. 
By denoting an orbital element set $\oe=(a,e,i,\Omega,\omega,M)$ where $a$ denotes the semi-major axis, $e$ the eccentricity, $i$ the inclination, $\Omega$ the longitude of the ascending node, $\omega$ the argument of perigee and $M$ the mean-anomaly, the MEM index assigns, for an admissible initial condition $\oe_{0}$, 
the maximal stretch of the orbital eccentricity over a finite time window $[0,t]$, $t \in \mathbb{R}_{+}$:
\begin{equation}\label{eq:Deltaecc}
\delta e(\oe_{0};t) 
= 
\max_{0 \le \tau \le t } e(\tau)- \min_{0 \le \tau \le t } e(\tau). 
\end{equation}
For short, and to follow conventional notations, we drop the initial datum $\oe_{0}$ and time $t$ from Eq.\,(\ref{eq:Deltaecc}) to simply note $\delta e$. 
The time window is problem dependent, and is chosen in accordance with the dynamical or physical timescales. 
The MEM  has been used in a variety of contexts and across distinct scales, ranging from the study of near-Earth space artificial satellite dynamics \cite{iGk16,emAl16,emAl18,jDa18,aRo19}
up to the resonant structure of satellites
around giant planets \cite{cCh22}, 
including the stability of planetary systems \cite{rDv04,vKo20,rAl21}.
Depending on the dynamical context, it might  be more appropriate to 
substitute Eq.\,(\ref{eq:Deltaecc}) by its analogue 
\begin{equation}\label{eq:Deltasma}
\delta a(\oe_{0};t) 
= 
\max_{0 \le \tau \le t } a(\tau)
- 
\min_{0 \le \tau \le t } a(\tau).  
\end{equation}
In using Eq.\,(\ref{eq:Deltasma}) over Eq.\,(\ref{eq:Deltaecc}), as \eg in the context of mean-motion resonances \cite{tGa16} or tesseral effects with the Earth \cite{cCo17}, a dynamical understanding of the problem is injected into the computed quantity (in both cases,  the semi-major axis is the variable primarily excited by the perturbation).
The $\delta e$ or $\delta a$ indices are most typically used in the context of dynamical maps where the scalars are colour-coded accordingly to their magnitudes, and computed over 2-dimensional grid of initial conditions (\ie a $\delta e$ or $\delta a$ heatmap is computed by  ``freezing'' $4$ variables in the initial datum $\oe_{0}$). 
As a matter of fact, the computation of the amplitudes (diameters) is a useful tool in delineating and visualising  dynamical structures  of multidimensional problems. 
This paper extends the MEM capabilities and reports an existing connection with chaos identification. 
In fact, we present a simple method to unravel sharply dynamical structures and chaotic seas from
the knowledge of orbits only and 
MEM like computations, yielding to a robust and sensitive non-variational chaos indicator.\\

We would like to stress that the reported method does not constitute a ``new branch'' of chaos indicator (among which it is customary to distinguish between frequential like methods (\eg \cite{jLa93}), or varational methods such as the the Fast Lyapunov Indicator (FLI) and variations \cite{cFr97,mFo02,rBa05}, the
MEGNO \cite{pCi03}, the SALI \cite{cSk01} or GALI \cite{cSk07} indices, to name but a few),  but is instead thoroughly  related to the  Lagrangian Descriptor (LD) framework in which the $M$-function assigning euclidean length\footnote{The LD framework does not rely exclusively  on the euclidean norm.  In this respect, $p$-norm like LDs \cite{cLo15}, LDs based on the actions in the context of Hamiltonians framework \cite{vGa22}, or time-free geometrical LDs for integrable problems \cite{rPO21} have been considered in several instances.} to orbits plays a central role \cite{jaMa09,cMe10,aMa13}. 
In fact, paraphrasing \cite{aMa13},  it is clear from Eqs.\,(\ref{eq:Deltaecc}) and (\ref{eq:Deltasma}) that $\delta e$ and $\delta a$ are actually positive quantities that accumulate along the trajectory, one of the cornerstone property behind LDs. 
LDs have been precious allies over the years for gaining dynamical understandings in a variety of contexts and range of fields, such as the detection of Lagrangian coherent structures in geophysical and oceanic flows (see \eg \cite{cMe10,jCu19I,jCu19II}), but also and especially in the field of reaction dynamics in theoretical chemistry, allowing to recover stable and unstable manifolds of normally hyperbolic invariant manifolds (NHIM) in a non-perturbative approach (see \eg \cite{gCr15,mFe17,aJu17,yNa21}).  
Recently, two non-variational chaos indicators have been proposed in concert from the $M$-function by \cite{jDa22-physD} and \cite{mHi22}. The present paper follows closely the steps of \cite{jDa22-physD} by recognising the resemblance of the $M$-function with MEM like quantities used in orbital settings.\\

The rest of the paper is organised as follows:
\begin{itemize}
	\item In Sect.\,\ref{sec:D}, we 
	introduce the MEM general counterpart in the setting of $n$ degree-of-freedom (DoF) action-angle Hamiltonian. The  introduced quantity corresponds to a diameter, $D$, measuring the maximal amplitude of the actions over a finite time window.   
	The diameter metric  is then studied over slices of initial conditions on paradigmatic models of resonances, supporting possibly chaotic motions. These include the integrable pendulum model, the modulated pendulum with a thin chaotic layer, and a two-waves Hamiltonian model where resonances overlap. 
	The analysis leads to a better understanding  of the dynamical drivers of the diameter. 
	\item In Sect.\,\ref{sec:DeltaD}, we introduce a scalar quantity encoding the regularity of the diameter metric. The index, denoted $\norm{\Delta D}$, is based on the evaluation of the second derivatives of $D$. This quantity is a sensitive and robust scalar able to detect sharply hyperbolic trajectories and  multi-resonant modes. It is proposed as a new non-variational chaos indicator. 
	\item In Sect.\,\ref{sec:PlanetaryApplication}, in the context of $N$-body simulations, we 
	apply our framework to the $2$ and $3$ bodies planetary problem. In this case, the diameter metric reduces to the original MEM like quantity given in Eq.\,(\ref{eq:Deltasma}), $\delta a$. We demonstrate the benefits of considering $\norm{\Delta \delta a}$ over $\delta a$ to restore the separatrices of the $2$ planet problem, and resonant templates of mean-motions configurations in the $3$ planet case. The $\norm{\Delta \delta a}$ offers details of the phase space otherwise unseen with the $\delta a$ quantity. 
\end{itemize}
We close the paper by summarising our results and main contributions. 

\section{Action-diameter applied to resonant models}\label{sec:D}
In order to extend the MEM  to the more general setting of $n$-\DOF Hamiltonian written in action-angle variables  $(I,\phi)$ in $\mathbb{R}^{n} \times \mathbb{T}^{n}$,
we find convenient to introduce the \textit{diameter} of an orbit. This follows the direction and  terminology employed by  \cite{nGu17} in the context of nearly-integrable maps (see complimentary results on the formalism developed in Appendix \ref{app:Discrete}). Let us denote by $\mathcal{D}$ a subset of $\mathbb{R}^{n} \times \mathbb{T}^{n}$. 
For a given initial condition $(I_{0},\phi_{0})$ in $\mathcal{D}$, we define the diameter metric as\footnote{We shall not consider here blowing-up trajectories in finite-time. We thus assume to deal with bounded observables, leading to a  finite diameter $D$. In the $N$-body simulations of Sect.\,\ref{sec:PlanetaryApplication}, escapes in finite time are not excluded. Nevertheless, we bypass this problem by using conditional exit loops during the numerical treatment of the equation of motions. This prevents the issue to happen. Equivalently, it amounts in some cases to consider the final time variable as a function of the initial datum, $t=t(I_{0},\phi_{0})$.} 

\begin{align}\label{eq:Dfunction}
D: \, \mathcal{D} \times \mathbb{R} &\to [0,+\infty), \notag \\
(I_{0},\phi_{0};t) &\mapsto 
D(I_{0},\phi_{0};t),
\end{align}
with
\begin{equation}\label{eq:Dvalue}
D(I_{0},\phi_{0};t)=
\norm{\big(\delta I_{1}(I_{0},\phi_{0};t),\cdots,\delta I_{n}(I_{0},\phi_{0};t)\big)}_{\infty},
\end{equation}
where each element of the set 
$\{\delta I_{j}(I_{0},\phi_{0};t)\}_{j=1}^{n}$,  represents the amplitude\footnote{Interestingly enough, we shall underline that similar definitions based on amplitudes in certain direction found applications in fluid mechanics for characterising mixing properties, see \cite{rMu14}.}  
of the considered action\footnote{
	Note that Eqs.\,(\ref{eq:Deltaecc}) and (\ref{eq:Deltasma}) are not based directly on actions, yet, the metrical orbital elements $(a,e,i)$ are simple function of proper actions such as the Delaunay elements $(L,G,H)$, for example. Thus, large variations in $a$ are equivalent to large variations in $L=\sqrt{\mu a}$, and, in the secular approximation where $L$ is a first integral, large variations in $e$ are equivalent to large variations in $G=L\sqrt{1-e^{2}}$.   
}: 
\begin{equation}\label{eq:DD}
\delta I_{j}(I_{0},\phi_{0};t) = 
\max_{0 \le \tau \le t} I_{j}(I_{0},\phi_{0},\tau)-
\min_{0 \le \tau \le t} I_{j}(I_{0},\phi_{0},\tau), \quad j=1,\cdots,n.
\end{equation}
To avoid notation burden, we shall not distinguish between the $D$ metric or its value (the diameter, noted $D$) provided an initial datum. It is clear from Eqs.\,(\ref{eq:Dvalue}) and (\ref{eq:DD}) that  the $D$ metric is an increasing function of time, accumulating a positive scalar along the orbit's history.  
We analyse now how the $D$ metric behaves on low-dimensional resonant models.\\

\subsection{Models and numerical settings}
We consider the following three archetypal resonant models  

\begin{align}
\left\{
\begin{aligned}
&\mathcal{H}(I,\phi)=\frac{I^{2}}{2}-\cos \phi, \quad (I,\phi) \in \mathcal{C}, \, \mathcal{C}=\mathbb{R} \times [0,2\pi], \notag \\ 
&\mathcal{J}(I,\phi,t)=\frac{I^{2}}{2}- \big(1 + \alpha \cos \epsilon t\big)\cos \phi, \quad (I,\phi,t) \in \mathcal{C} \times \mathbb{R}, \notag \\
&\mathcal{K}(I,\phi,t) = \frac{I^{2}}{2} - (
\alpha_{1} \cos(\phi-t) + \alpha_{2}\cos(\phi+t) ), \quad (I,\phi,t) \in \mathcal{C} \times \mathbb{R},
\end{aligned}
\right.
\end{align}
where $\epsilon \ll 1$, $\alpha < 1$, $\alpha_{1}$, $\alpha_{2}$ 
are real positive  parameters. 
For each model, described further in the subsequent, we compute values of $D$ over chosen slices of initial conditions. 
In order to evaluate Eq.\,(\ref{eq:Dfunction}), we have fixed the time window to $\mathcal{T}=[0,500]$.
The non-autonomous models have been converted into $2$-\DOF autonomous models by extending the dimension of the phase space through the introduction of canonical variables $(J,\tau)$, where $\dot{\tau}$ has a trivial dynamics. 
As the ``dummy'' action $J$ has no dynamical relevance, the diameter $D$ is computed only by monitoring the action $I$ (\ie we compute a one dimensional diameter corresponding to the amplitude of the action $I$). 
Although the pendulum model $\mathcal{H}$ is $1$-\DOF with a phase space easily described by the level-set method (see further discussions on analytical properties of $D$ in Appendix \ref{app:AnalProof}), all the corresponding flows have been numerically approximated using numerical solvers. 


\subsection{Application to the integrable pendulum}\label{subsec:pend}
The well-known phase space of the $1$-\DOF pendulum model $\mathcal{H}$ on ${\mathcal D} \times [-\pi,\pi]$, ${\mathcal D} \subset \mathbb{R}$, obtained by the level-set method is shown in the left frame of Fig.\,\ref{fig:fig1}. 
The phase space contains the three fixed points, the stable equilibrium at the origin $(\phi,I)=(0,0)$ and the fixed points $(\phi,I)=(\pi,0)=(-\pi,0)$. 
The level curve associated to the energy $E=1$ of the hyperbolic equilibrium, \ie the separatrix (shown as a red curve), divides orbits of the phase space with distinct qualitative behaviours. 
Within the cat-eye domain ($E < 1$), the phase space is foliated by librational curves, whilst outside the cat-eye region ($E > 1$), circulational tori enclose the cylinder.  
The half-width $\delta I$ of the resonance, \ie the distance  between $I=0$ and the apex of the separatrix, satisfies 
\begin{equation}
\mathcal{H}(0,\delta I)=1.
\end{equation}
Solving this last equation for $\delta I$, one find $\delta I=2$ leading, by symmetry, to the full resonant width  $\Delta I = 2 \delta I=4$. 
The right companion panel of Fig.\,\ref{fig:fig1} shows the landscape
of the $D$ metric computed as a function of the initial action $I$ for the fixed angle $\phi=0$ (corresponding to the vertical dashed blue line in the phase space of the pendulum). 
For $I \ge 0$, the $D$ function grows linearly within the librational domain up to the apex of the separatrix (corresponding to the action $I= 2$). 
The crossing of the separatrix is materialised by a sudden loss of continuity in the $D$ function. At this point, $D$ is only left-continue and reaches its maximal value, $D(2)= \Delta I =4$. 
For $I > 2$, the graph decreases monotonically.  A similar picture is obtained by symmetry for $I < 0$. 
From the numerical estimation of $D$, one infer that $D$ is not differentiable at $I=0$, and $I=\pm 2$ (where the function is not even continuous). 
Those properties of the $D$ metric are proven analytically in Appendix \ref{app:AnalProof}.  

\begin{figure}
	\centering
	\includegraphics[width=0.9\linewidth]{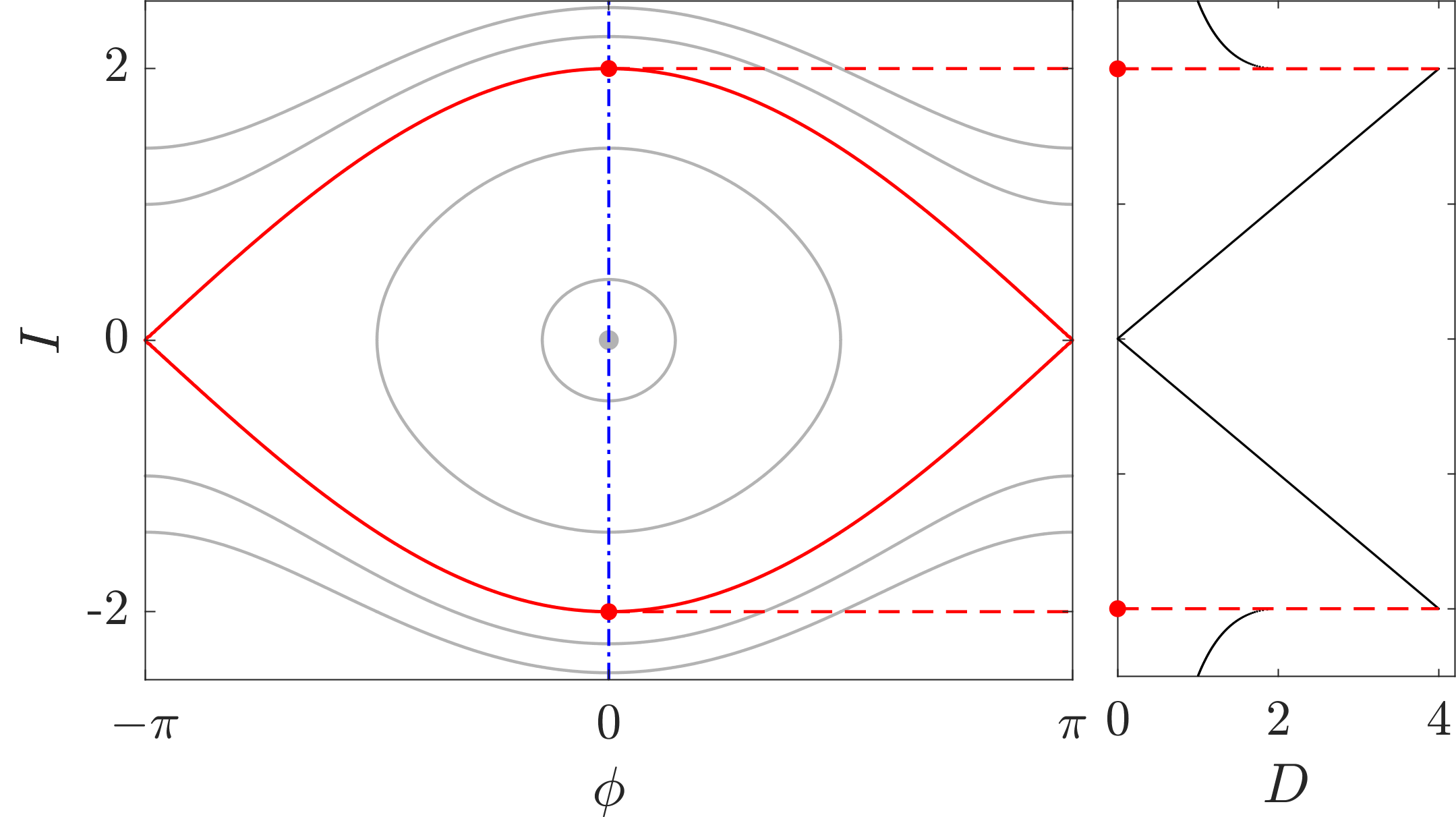}
	\caption{
		(Left) Phase space of the pendulum model $\mathcal{H}$ with the cat-eye separatrix topology (red).
		(Right) Landscape of the diameter $D$ computed over the dashed blue line of initial conditions for fixed $\phi=0$ and varying $I$ ranging the circulational and librational domains. The red horizontal lines materialise the actions corresponding to the separatrix crossing. At those values, $D$ is discontinue. 
	}
	\label{fig:fig1}
\end{figure}

\subsection{Application to slow chaos of a modulated pendulum}
The $2$-\DOF autonomous version of $\mathcal{J}$, still denoted $\mathcal{J}$,
reads 
\begin{equation}
\mathcal{J}=
\frac{I^2}{2}+\epsilon J
-
(1 + \alpha \cos \tau) \cos \phi.
\end{equation}
As the time variable $\tau$ is slow ($\dot{\tau} =\epsilon$, $\epsilon \ll 1$), this model is paradigmatic of slow chaos where $3$ resonances are $\epsilon$-close \cite{yEl93}. 
In fact, using trigonometrical identities, $\mathcal{J}$ might be written as
\begin{equation}
\mathcal{J}=
\frac{I^2}{2}+\epsilon J-
\big(
\cos \phi
+\frac{\alpha}{2} \cos(\phi-\tau)
+\frac{\alpha}{2} \cos(\phi+\tau)
\big),
\end{equation}
where the $3$ harmonics are clearly apparent, and $\epsilon$ apart. Indeed,  
using Hamilton's canonical equations, one sees that the centres of the resonances 
$\dot{\phi}=0$, $\dot{\phi}-\dot{\tau}$
and $\dot{\phi}+\dot{\tau}$ correspond respectively to the actions values $I=0$, $I=\epsilon$, $I=-\epsilon$. 
Iterations of the stroboscopic mapping computed for the numerical values $\alpha=0.25$, $\epsilon=0.1$,  and obtained by projecting in the $(I,\phi)$ plane snapshots of the flow at times $t$ such that $\epsilon t = 0 \mod 2\pi$, 
are shown in the left panel of Fig.\,\ref{fig:fig2}. 
The $(I,\phi)$ phase space contains predominantly regular curves, and a stochastic layer surrounding the unperturbed separatrix. In fact,   
in the limit $\epsilon \to 0$, the outer and inner boundaries of the chaotic sea can be related to instantaneous separatrices associated to integrable approximations of $\mathcal{J}$ by freezing the time related variable \cite{yEl93}.   
The $D$ landscape, computed for $(J,\phi,\tau)=0$, is similar to the landscape of the pendulum model. 
The most noticeable difference occurs for the range of actions crossing the hyperbolic layer, for which the $D$ metric looses its regularity, as seen in the two-scales plot in the right panel of Fig.\,\ref{fig:fig2}.  

\begin{figure}
	\centering
	\includegraphics[width=0.9\linewidth]{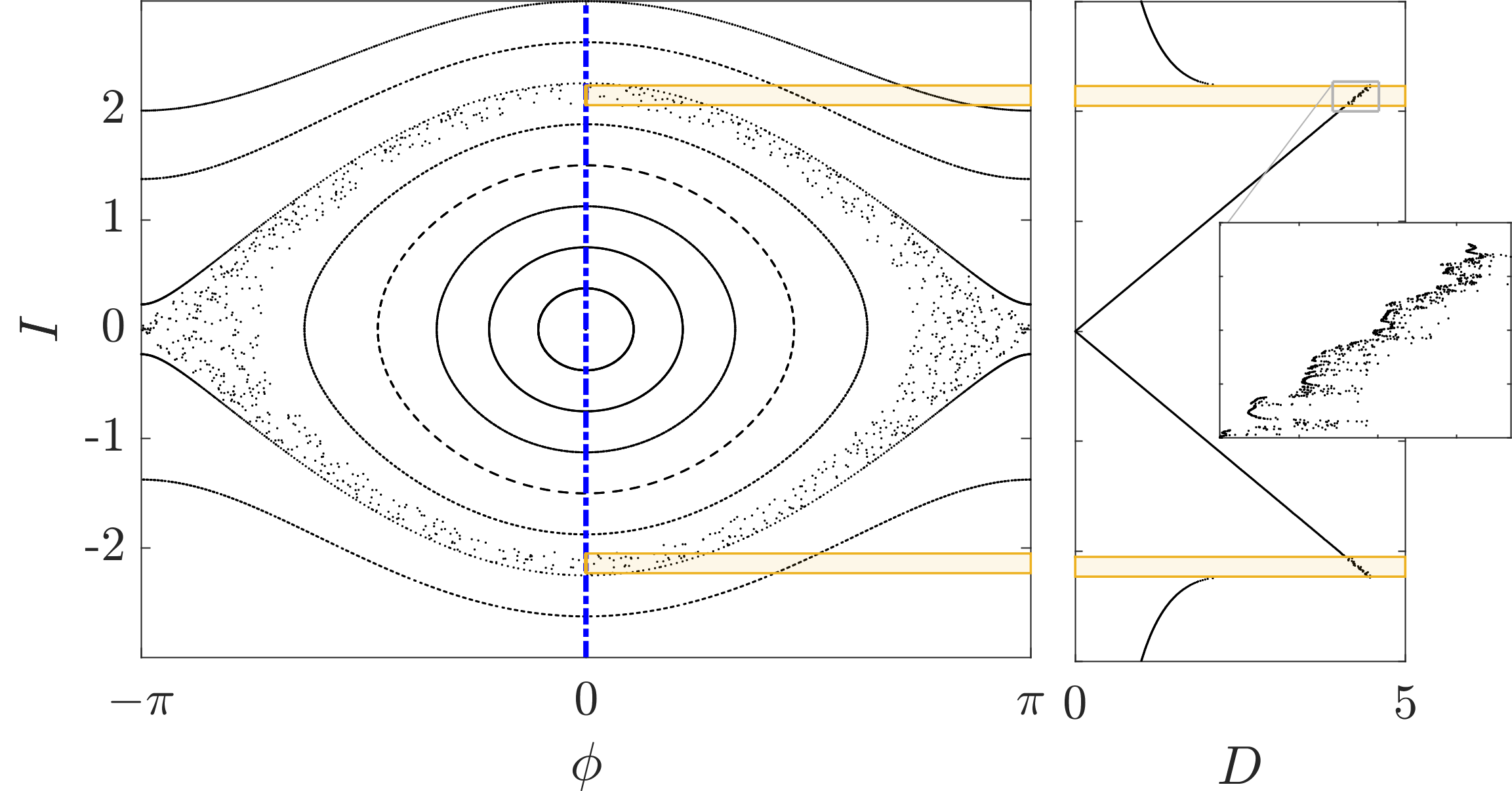}
	\caption{
		(Left)
		Phase space of the modulated pendulum $\mathcal{J}$, using $\alpha=0.25$ and $\epsilon=0.1$, 
		obtained by iterating the period-map.
		(Right) Landscape of the diameter $D$ computed over the dashed blue line of initial conditions.
		The range of the chaotic layer are delineated with the slightly transparent filled areas. It establishes a link between the loss of regularity of $D$ and the crossing of the hyperbolic region. 
		The inlaid panel shows a zoomed in portion of $D$ when crossing the chaotic layer. 
	}
	\label{fig:fig2}
\end{figure}

\subsection{Application to resonance overlap}
The $2$-\DOF autonomous counterpart of $\mathcal{K}$ reads
\begin{equation}
\mathcal{K}
=
\frac{I^{2}}{2}+J-
(
\alpha_{1}
\cos(\phi-\tau)+
\alpha_{2}\cos(\phi+\tau)
).
\end{equation}
When $\alpha_{1}=0$ or $\alpha_{2}=0$, we recover the integrable Hamiltonian of the pendulum using an ad-hoc canonical change of variables.  
The phase space then contains a single cat-eye resonance centred around either $c_{1}=-1$ or $c_{2}=1$, with half-widths $\delta_{2}=2\sqrt{\alpha_{2}}$ or $\delta_{1}=2\sqrt{\alpha_{1}}$ respectively. Whenever both $\alpha_{1}$ and $\alpha_{2}$ are different from zero, the $2$-\DOF Hamiltonian is no longer integrable \cite{dEs81}.  
The resonance overlap  parameter, also called \textit{stochasticity parameter} \cite{bCh79,jMe07}, reads
\begin{equation}
s=\frac{\delta_{1}+\delta_{2}}{\vert c_{2} - c_{1}\vert}=\sqrt{\alpha_1} + \sqrt{\alpha_2}.
\end{equation}
In our numerical setting, we assign to each resonant eye the same dynamical weight with $\alpha_{1}=\alpha_{2}=1/5$, leading to $s\approx 0.89$.  As $s$ is close to $1$, the resonances overlap significantly and macroscopic chaos is expected. 
Iterations of the stroboscopic map, obtained by projecting the flow in the $(I,\phi)$ plane for times $t$ such that $t = 0 \mod 2\pi$, are shown in the left panel of Fig.\,\ref{fig:fig3} for $I \ge 0$. The phase space contains a large connected chaotic sea. The remnants of the librational domains contains several chains of periodic orbits surrounded by thin chaotic layers.  The observation made before on the regularity of the $D$-metric when crossing hyperbolic tangles is made more evident, as illustrated in the right panel of Fig.\,\ref{fig:fig3}. 
The landscape of the $D$ metric contains the characteristic V-shape already observed when crossing elliptical regions, and becomes irregular for actions leading to hyperbolic motions. 

\begin{figure}
	\centering
	\includegraphics[width=0.9\linewidth]{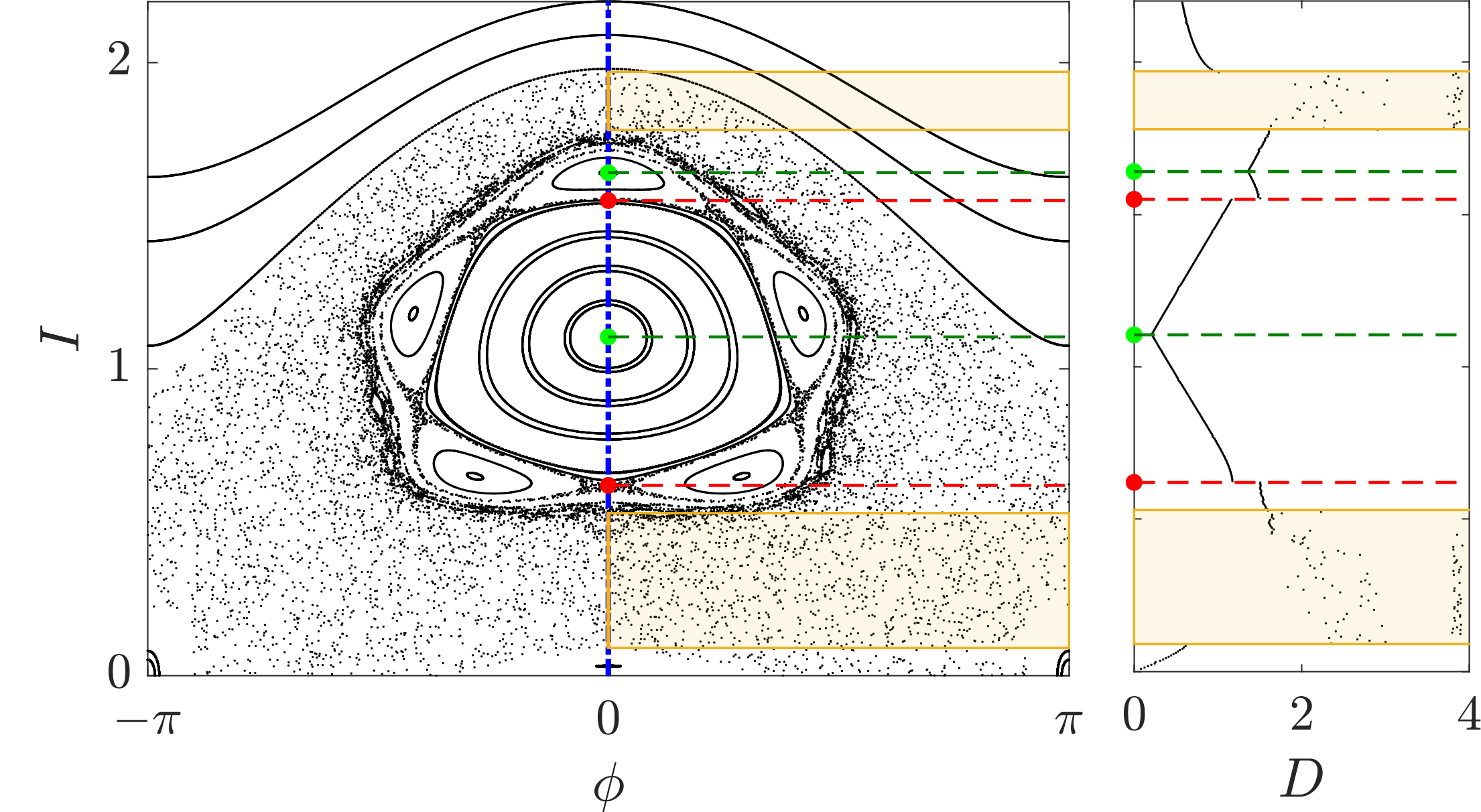}
	\caption{
		(Left)
		Phase space of the two-waves Hamiltonian $\mathcal{K}$, with $\alpha_{1}=\alpha_{2}=1/5$, obtained by iterating the period-map.
		(Right) Landscape of the diameter $D$ computed over the dashed-blue line of initial conditions.
	}
	\label{fig:fig3}
\end{figure}

\section{The $\norm{\Delta D}$ indicator}\label{sec:DeltaD}
\begin{figure}
	\centering
	\includegraphics[width=0.7\linewidth]{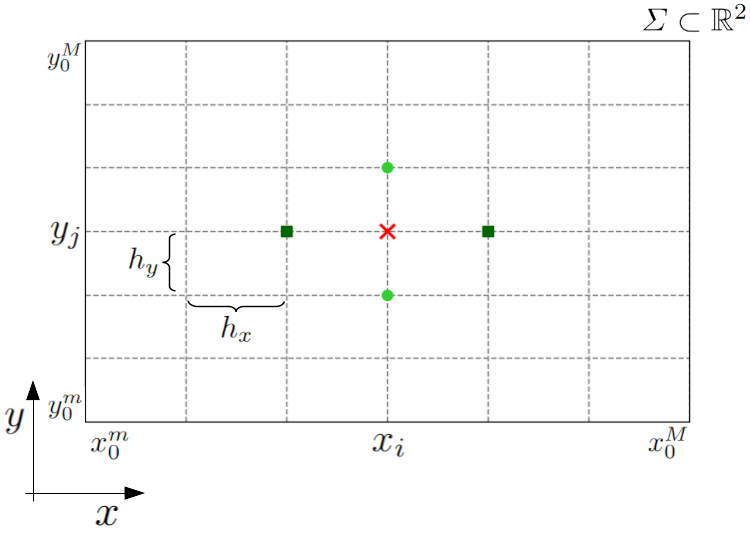}
	\caption{
		Schematic representation of the approach taken to estimate the $\norm{\Delta D}$ index through the second derivatives appearing in Eq.\,(\ref{eq:LapD}) when the computation is performed on a $2$-dimensional section $\Sigma$.
		For a point $(x_{i},y_{j}) \in \mathring{\Sigma}$ (cross), besides the value of the diameter at the point $(x_{i},y_{j})$ itself, 
		the second derivatives are estimated using the $4$ neighbouring values at the points 
		$x_{\pm}=x_{i} \pm h_{x}$ (diamonds) and 
		$y_{\pm}=y_{j} \pm h_{y}$ (circles), accordingly to Eq.\,(\ref{eq:DerivativeFD}). The couple $(h_{x},h_{y})$ define the resolution of the Cartesian mesh obtained as 
		$h_{x}=(x_{0}^{M}-x_{0}^{m})/N_{x}$
		and 
		$h_{y}=(y_{0}^{M}-y_{0}^{m})/N_{y}$. 
		Typical values to produce a resolved map are $N_{x}=N_{y}=500$.   
	}
	\label{fig:fig4}
\end{figure} 
As computationally just observed, the diameter metric encapsulates  signatures of relevant dynamical information, however, the latter are not encoded into the final value of $D$ but rather
in the regularity of the application. This is reminiscent of properties of the $M$-function \cite{aMa13}.
This observation, however, conflicts with the current use of the diameter to visualise dynamical structures, especially in the context of stability maps. Albeit we are not able to provide a general proof, based on the former observations, we conjecture the $D$ metric to be non differentiable when crossing transversally hyperbolic domains. 
This property offers the possibility of delineating sharply hyperbolic domains by quantifying instead the regularity of the application. 
In this respect, following the same strategy of the frequency analysis \cite{jLa93}, 
we find convenient to introduce the second-order derivative based quantity
\begin{equation}\label{eq:LapD}
\norm{\Delta D(I_{0},\phi_{0};t)} = 
\sum_{i=1}^{n} 
\Big\vert
\partial_{I_{0}^{(i)}I_{0}^{(i)}}^{2}
D(I_{0},\phi_{0};t)
\Big\vert,
\end{equation}
where $I_{0}^{(i)}$ denotes the $i$-th component of the initial datum $I_{0}$.\\

\begin{reeemark}
	In the following, the diameters $D$ are estimated numerically using discretised domains of initial conditions.  
	We  compute
	Eq.\,(\ref{eq:LapD}) using central differences. 
	For the sake of simplicity, assume we are evaluating $D$ as a function of $I \in \mathcal{I}=[i_{\min},i_{\max}] \subset \mathbb{R}$. We evaluate 
	$D$ for each action $I_{j}=i_{\min}+jh$, with $h=(i_{\max}-i_{\min})/N$, $N$ being a sufficiently large natural number fixing the resolution of the mesh.  
	The numerical approximation of Eq.\,(\ref{eq:LapD}), for points in $\mathring{\mathcal{I}}$ reads
	\begin{equation}\label{eq:DerivativeFD}
	\frac{\partial^{2}D(I_{j})}{\partial I^{2}}   \simeq \frac{D(I_{j+1})+D(I_{j-1})-2D(I_{j})}{h^{2}}.
	\end{equation}
	This procedure is generalised to higher dimensional domains of initial conditions, \eg for sections $\Sigma \subset \mathbb{R}^{2}$, as schematically represented in Fig.\,\ref{fig:fig4}.\\
\end{reeemark}

Fig.\,\ref{fig:fig5} presents a heatmap of the diameter associated to the pendulum model (a similar computation is performed by \cite{xRa15}), and its associated $\norm{\Delta D}$ landscape computed for varying $I$ and fixed $\phi=0$ (dashed vertical line). 
The points leading to the largest diameters correspond to librational orbits whose energy approaches the energy associated to the separatrix (where $D \sim 4$).  
This leads to a map where the vicinity of the separatrix (within the librational domain) takes approximately the same diameter values. 
Instead, the $\norm{\Delta D}$ landscape reacts as sharp Dirac impulses for the actions corresponding to the separatrix ($I=\pm 2$). The values taken by 
$\norm{\Delta D}$ 
are in fact several orders of magnitude larger than the values associated to the diameters of others librational or circulational orbits. 
The only exception appears in the vicinity of the stable equilibria at $I=0$, where a Dirac 
pic can also be distinguished, with a lower amplitude though.
This observation is consistent with the $V$-shape of the diameter near the stable equilibria (confer Fig.\,\ref{fig:fig1}) and  the analytical estimates presented in Appendix \ref{app:AnalProof}. \\

The next section  demonstrates further the benefits of supplementing the $D$ analysis with the  $\norm{\Delta D}$ index in a practical planetary context.

\begin{figure}
	\centering
	\includegraphics[width=0.85\linewidth]{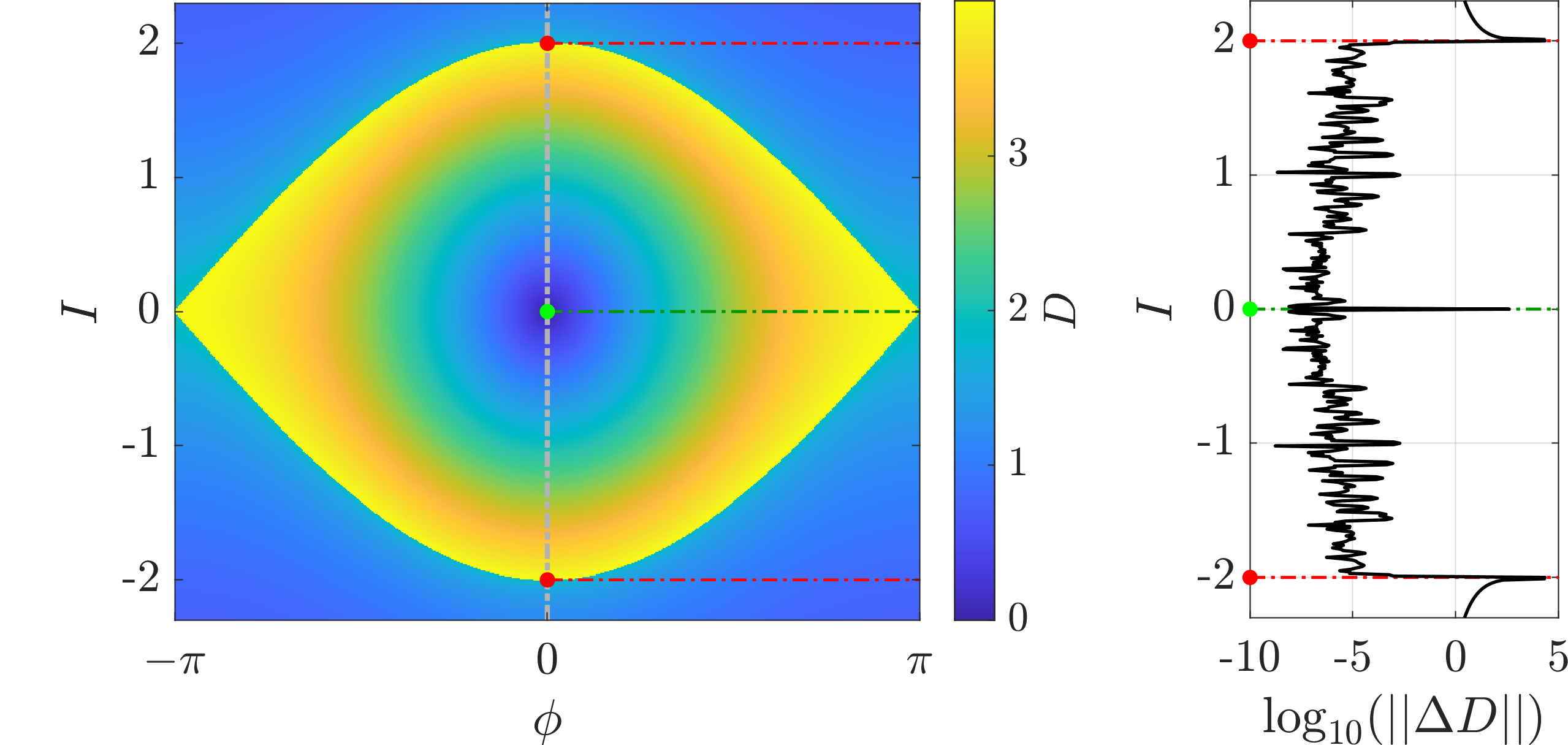}
	\caption{
		Diameter stability map associated to the pendulum model and its corresponding $\norm{\Delta D}$ landscape computed along the line $\phi=0$ (black dashed black line).
		The $\norm{\Delta D}$ metric allows to sharply detect the separatrix crossing, contrarily to the diameter which take similar values in its whole  neighbourhood (with $D \simeq 4$). The dashed red lines at $I=\pm 2$ and the dashed green line at $I=0$ in the landscape materialise the locations of the hyperbolic and stable equilibria respectively, where $\norm{\Delta D}$ now reacts as net impulses.
	}
	\label{fig:fig5}
\end{figure}

\section{Application to planetary problems}\label{sec:PlanetaryApplication}
We now return to the early roots of the MEM grounded in gravitational problems where we apply our apparatus into an $N$-body framework, namely the 2 and the 3 planet cases.
This section demonstrates the ability of the $\norm{\Delta D}$ indicator to reveal separatrices and thin details of resonant webs in the context of dynamical maps. 

\subsection{Generalities and numerical setups}
Our dynamical system consists in $N_{\rm pl}$ coplanar planets with masses $m_{i}$, $i=1,\dots,N_{\rm pl}$, 
orbiting a central star with mass $M_\star = 1 \, {\rm M}_\odot$ where $M_\star \gg m_i$. 
We study the orbital evolution of the planets with a set of orbital elements $\oe_i$ as defined in the introduction, with the sub-index referring to planet $i$ in the system. 
In terms of the modified Delaunay canonical variables \cite{jLa91}, 
\begin{align}
\left\{
\begin{aligned}
& L_i=m_i \sqrt{\mu_i a_i}, \hspace{1.7cm}\quad \lambda_i=M_{i} + \omega_{i}, \\
& S_i = L_i \left(1-\sqrt{1-e_i^2}\right), \quad -\varpi_i=-\omega_{i}, 
\end{aligned}
\right.
\end{align}
the system is described by the Hamiltonian function
\begin{equation}\label{eq:ham}
{\mathcal H} = - \sum_{i = 1}^{N_{\rm pl}} \frac{\mu_i m_i^{3}}{2L_i^2} + 
\sum_{\bf k } c_{\bf k }(\bold{L},\bold{S}) \cos \left( \bold{k} \cdot \boldsymbol{\phi} \right)
\end{equation}
where $\mu_i = {\mathcal G} (M_\star + m_i)$, ${\mathcal G}$ denotes the gravitational constant and  $(\bold{L},\bold{S})$ denotes the vectors whose components 
are $(L_{i},S_{i})$, similarly to $\boldsymbol{\phi}=(\boldsymbol{\lambda},-\boldsymbol{\varpi})$, and $\bold{k} \in \mathbb{Z}_\star^{2N_{\rm pl}}$. 
The first term in Eq.~\eqref{eq:ham} is the integrable part that refers to the unperturbed Keplerian motion of the planets around the central star. The second term is the perturbing function which accounts for the interactions between the planets. The $c_{\bold{k}}$ coefficients in the perturbing function depend on the Laplace coefficients and can be computed from the expressions given in \cite{cMu99}. Our system is thus a $2 N_{\rm pl}$-\DOF problem. \\

In a general manner, a 3-planet mean-motion resonance (MMR) can be expressed in planetary orbital parameters as
\begin{equation}
k_1 n_1 + k_2 n_2 + k_3 n_3 = 0,
\label{eq:res0}
\end{equation}
where $n_i$ is the mean motion of planet $i$, and $
k=(k_{1},k_{2},k_{3}) \in \mathbb{Z}^{3}_{\star}$. 
The coefficient $s = |k_{1} + k_{2} + k_{3}|$ is the order of the resonance. The order of a 2-body MMR is known to be the order at which the eccentricity of the bodies appear in the coefficients of the expansion of the perturbing function. In the case of a 3-body resonance the property still holds, and thus we can separate 3-body resonances into zeroth order ($q = 0$), and non-zero order ($q \ne 0$). 
The particular case of the 2-planets resonant dynamics is straightforward to describe from the previous expression. A resonance between planets $m_1$ and $m_2$ can be described from Eq.\,(\ref{eq:res0}) when $k_3=0$, a resonance between planets $m_2$ and $m_3$, corresponds to $k_1=0$, while the last possible 2-planet MMR is between the non adjacent planets, when $k_2=0$.
In the following, we perform $N$-body simulations using the code as described by \cite{cBe12}. As a general description, for each case we choose a representative plane, and generate a grid of initial conditions that we integrate for a fixed amount of time. In both scenarios, all angles are set to zero. 


\subsection{The 2-planet case}
Although the pendulum model described in Sect.\,\ref{subsec:pend} allows for analytical estimations of many of the properties of the MMRs, the Second Fundamental Model of resonance of \cite{jHe83} is more adequate to reproduce the dynamical features in the 2-planet case for planets in circular orbits. 
Recalling that two planets are in a $(p+q)/p$ MMR with $p,q \in \mathbb{Z}_{\star}$ if the mean motions $n_1$ and $n_2$ of the planets satisfy $(p+q)n_2 - p n_1\sim 0$ and at least one of the associated resonant angles $\phi_{1,2} = (p+q)\lambda_2 -p \lambda_1 - q\varpi_{1,2}$ librates around a fixed value, 
the resonant dynamics of two planets on eccentric orbits in a MMR of arbitrary order $q$ can be reduced to a $1$-DoF Hamiltonian and described with a pendulum like structure (see \eg \cite{wSe84,sHa19}). 
Here we will not write explicitly this procedure, but let us mention that it is possible to derive analytically the location of the fixed points and the separatrix. 
Instead, we perform brute-force $N$-body simulations to estimate the main features of the system. 
MMR affects primarily the semi-major axis observables, and we thus rely on estimating $\delta a$, such that Eq.\,(\ref{eq:Dvalue}) becomes
\begin{equation}
\delta a = 
\max\big(
\delta a_{1}, 
\delta a_{2}
\big).
\end{equation}
Focusing on the vicinity of the first-order $n_{1}/n_{2} = 2/1$ MMR, Fig.\, \ref{fig:fig6} shows composite results of both the $\delta a$ and $\norm{\Delta\delta a}$
analyses in the representative $(a_1/a_2, e_1)$ plane, constructed over a grid of $500 \times 500$ initial conditions integrated for $5\times 10^3$ years. The planetary masses are chosen equal to $0.05$ and $0.1$ masses of Jupiter, ${\rm m_{Jup}}$.

\begin{figure}
	\centering
	\includegraphics[width=\textwidth]{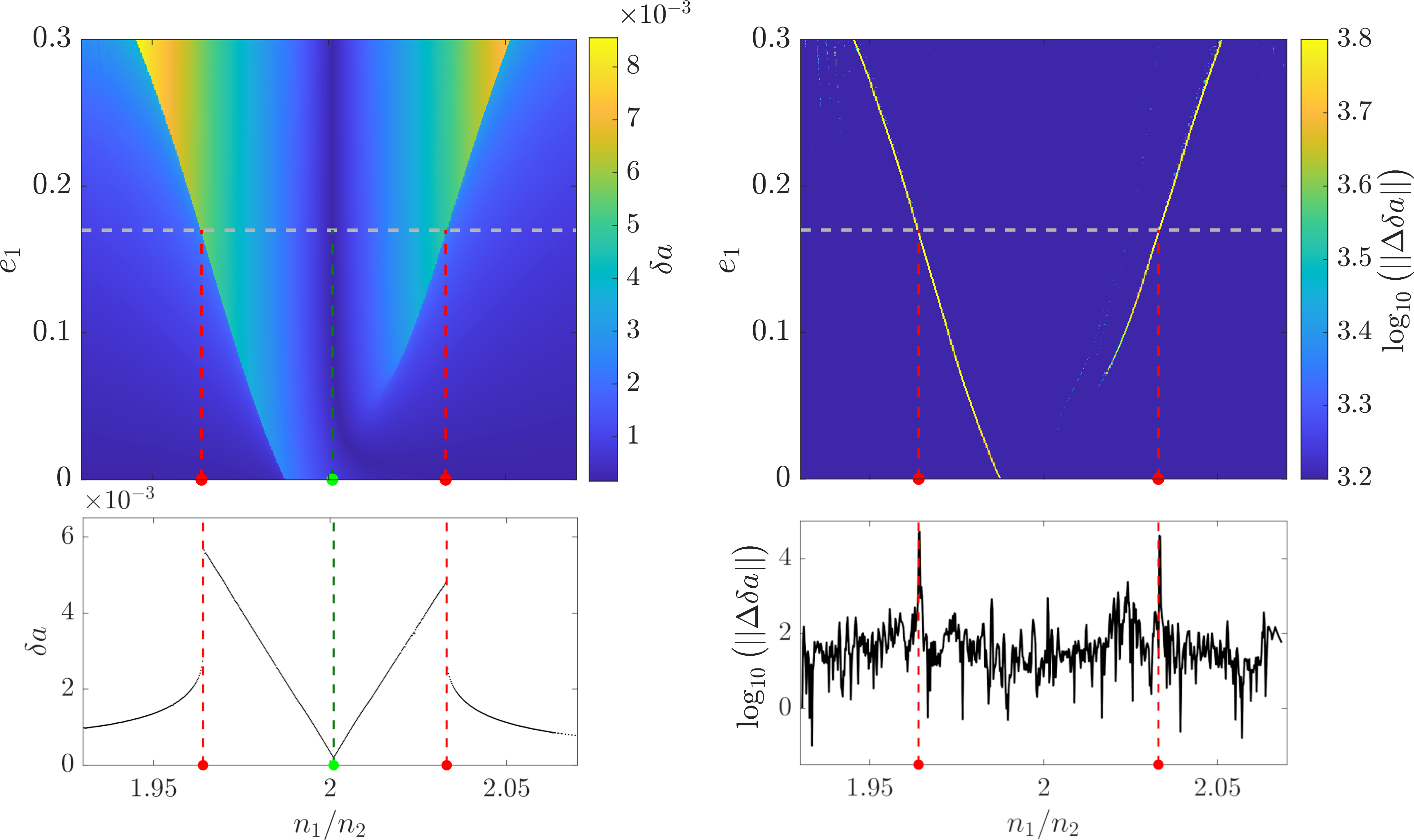}
	\caption{(Top left) Dynamical map using  $\delta a$  for the 2/1 MMR in the representative plane $(n_1/n_2, e_1)$ of initial conditions. Planetary masses are $m_1 = 0.05 \, {\rm m_{Jup}}$ and $m_2 = 0.1 \, {\rm m_{Jup}}$ orbiting a central star with mass $M_\star = 1 \, {\rm M}_\odot$. 
		Angular variables are initially chosen equal to zero. The pericentric branch of zero-amplitude solutions appears as dark blue, while the regions of maximum variation of the semi-major axis appear in light colours. Total integration time is $10^4 \, {\rm yr}$. (Bottom left) Landscape of the $N$-body integration of  $\delta a$ computed over the dashed line of initial conditions at fixed eccentricity $e_1 = 0.17$. 
		(Top right) Corresponding $\norm{ \Delta \delta a }$ map which clearly identifies the separatrices. 
		(Bottom right) Corresponding landscape of $\norm{\Delta \delta a}$.}
	\label{fig:fig6}
\end{figure}

The top left panel of Fig.\,\ref{fig:fig6}, which reproduces Fig.\,2 of \cite{xRa17}, identifies the characteristic V-shape of the resonant structure of the 2-body problem, although the diameter takes similar values in the neighbourhood of the separatrix (similarly to Fig.~\ref{fig:fig5}).
The benefits of considering $\norm{\Delta \delta a}$ over $\delta a$ are made evident in the dynamical map of the right-hand panel. The V-shape of the separatrix is now  undoubtedly identifiable. Each dynamical map comes with its corresponding landscape computed over the dashed line of initial conditions after varying $n_{1}/n_{2}$ but frozen $e=0.17$.  
The obtained $\delta a$ landscape is analogue to the one obtained for the integrable pendulum (recall  Fig.\,\ref{fig:fig1}) and contains $3$ points where $\delta a$ looses its regularity. 
Two of them correspond to the separatrix crossing (red vertical lines), and the latter corresponds to the crossing of the pericentric branch (green vertical lines at $n_{1}/n_{2} \sim 2$). 
Whilst $\delta a $ becomes singular at this point, 
the numerical values of 
$\norm{\Delta \delta a}$ does not permit to distinguish it sharply. In fact, the $\norm{\Delta \delta a}$ landscape contains Dirac pics only 
for the period ratio corresponding to the separatrix crossing. 
Consequently, the pericentric branch (or family of stable solutions) is not identifiable.

\subsection{The 3-planet case}
We now turn our attention to the 3-planet dynamics. The dynamics of these resonant system is governed by the Second Fundamental Model of Resonance. Specific 3-body systems were studied by different means.  Let us only mention some of those like the asteroids in the Solar System \cite{dNe98} and exoplanetary systems like TRAPPIST-1 \cite{mGi17,rLu17} or TOI-178 \cite{aLe21}.
A model for zeroth-order 3-body resonances was provided by \cite{aQu11} who also derived a resonance overlap criterion. More recently, \cite{aPe21} generalised the result and has proposed an integrable model for first-order MMRs ($s=1$). From a numerical perspective, the detailed analysis of the resonant structure has been provided in both \cite{cCh18} and \cite{aPe21}.

We retake here the route of \cite{cCh18} and discuss the resonant template through $\delta a$ computations. We consider a system of three equal mass planets with $m_i = 30$ Earth masses, orbiting a Sun-like star. We adopt the $(n_1/n_2,n_2/n_3)$ representative plane and estimate $\delta a$ and its Laplacian $\norm{\Delta \delta a}$ over a grid of $500 \times 500$ initial conditions for $10^4$ years (or equivalently, $10^4$ orbits of the outer body since it is fixed at $a_3 = 1 \, {\rm au}$).  
The main features observed in the $\delta a$ map shown in the to panel of Fig.\,\ref{fig:fig7} have already been described by \cite{cCh18,aPe21}. Briefly, the main vertical stripes represent interactions between inner and middle planet, while horizontal lines show the MMRs between middle and outer bodies, and the diagonal curves with negative slope represent commensurabilities between innermost and outermost planets. The diagonal like lines (with positive slopes) are 3-planet zero- and first-order commensurabilities. The connected region of orbits with large $\delta a$ correspond to either to collisions or escapes.  
The comparison between the $\delta a$ and the $\norm{\Delta \delta a}$ map (bottom panel of Fig.\,\ref{fig:fig7}) highlights the striking advantages in considering the $\norm{\Delta \delta a}$ indicator. 
Besides the main structures detected by the $\delta a$ index, the map contains much more details on secondary structures, and is able to reveal sharply the complex topology of the interacting resonances.  \\

\begin{reeemark}
	We underline that the level of details
	offered by the $\norm{\Delta\delta a}$ map is similar to  the analysis performed using the well-established \texttt{MEGNO} chaos indicator \cite{pCi03}, confer Fig.\,2 of \cite{aPe21}.\\
\end{reeemark}

\begin{reeemark}
	The apparent ``flatness'' of the $\delta a$ map in Fig.\,\ref{fig:fig7} was bypassed by \cite{cCh18} using an analytical procedure, aiming at removing short-period oscillation terms of the semi-major axis (equivalently, the Delaunay variables $L$), by isolating and recognising the purely periodic components of the disturbing functions. In doing so, the long-term features recognisable in the map were enhanced. 
	Here, the numerical $\norm{\Delta \delta a}$ seems to outperform the analytical filtering.  
	\\
\end{reeemark}

\begin{figure}
	\centering
	\includegraphics[width=0.9\textwidth]{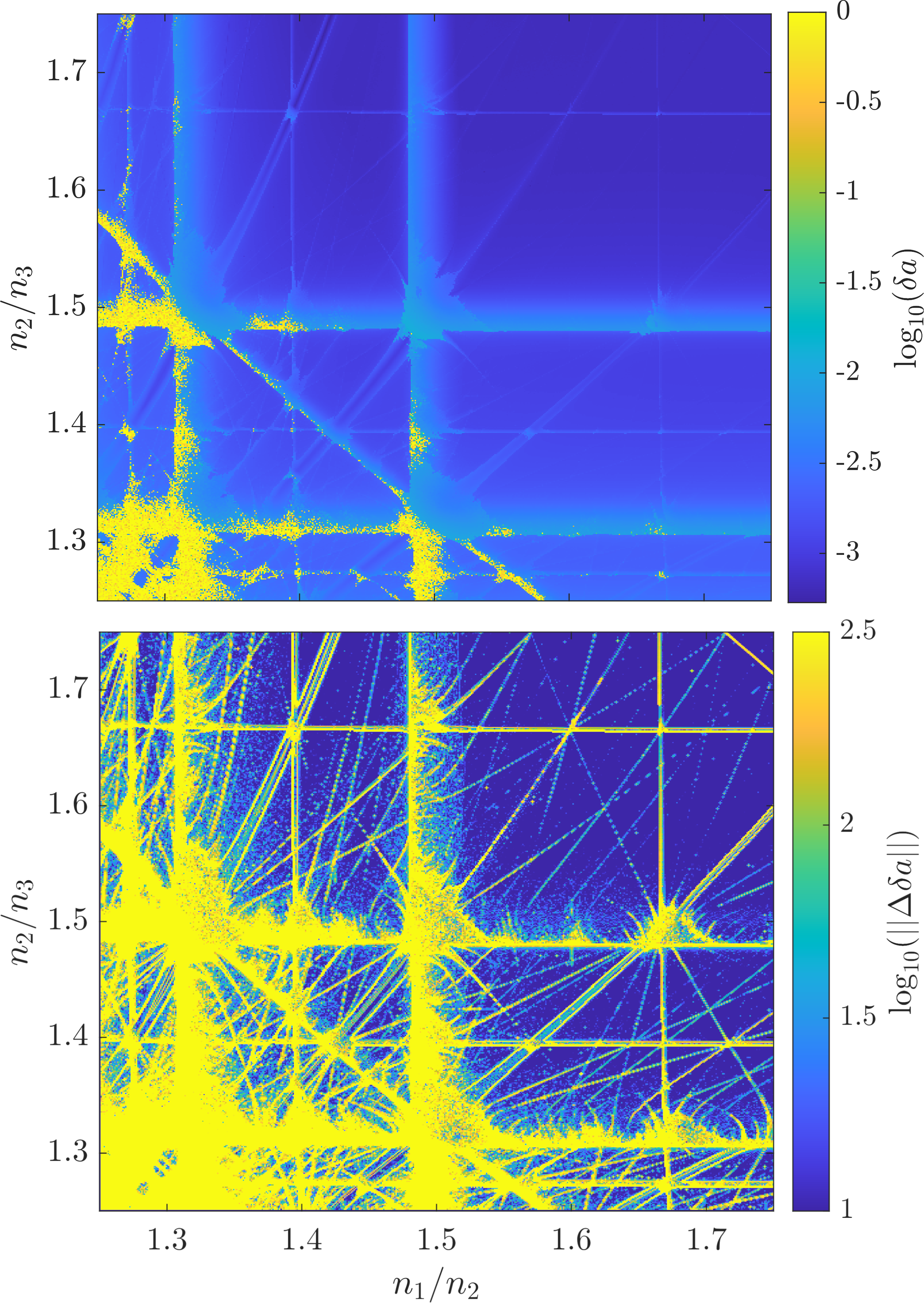}
	\caption{ Dynamical maps for the 3-planet case in the $(n_1/n_2,n_2/n_3)$ plane with a central star with mass $1 \, {\rm M}_\star$ and three planets with equal masses $m = 30 \, {\rm m_\oplus}$. The grid consists in $500 \times 500$ initial conditions propagated over $10^4$ periods of the outer planet (placed at $a_3 = 1 \, {\rm au}$).
		The $\delta a$ map (top) underestimate the overall resonant and chaotic architecture, and fails in 
		recognising thinner secondary structures as they appear with the $\norm{\Delta \delta a}$ analysis (bottom).} 
	\label{fig:fig7}
\end{figure}

\section{Conclusions}\label{sec:conclusions}
Dynamics plays a pivotal role in a wide range of scientific and engineering efforts. In planetary sciences, it covers evolution of Solar System's minor bodies, exoplanetary systems, and ultimately conditions for habitability. Determining the architecture of multi-planetary systems is one of the cornerstones for understanding planet formation and evolution of extrasolar systems as well as our own. The characterisation of extrasolar planets via their dynamics further supplies us with lots of clues hidden in the formation process. 
Having efficient methods to detect and visualise resonant structures 
is a key advance in the field. 
This work has extented and complemented the MEM capabilities by introducing a scalar value  inflating hyperbolic structures from their computations. The most important contributions of our work are summarised in the following: 
\begin{enumerate}
	\item Celestial mechanicians and astrodynamicists have been computing Lagrangian Descriptors like quantities for almost $2$ decades. We have established an analogy between the $M$-function, commonly employed in ocean and reaction dynamics, with diameters like quantities such as the MEM employed in gravitational dynamics. 
	\item We have introduced a non-variational  dynamical indicator from MEM like computations.
	The index, complementing further the MEM information, allows to enhance the visualisation of global structures, and improves MEM maps which tend to be ``flat.'' 
	The key point relies in quantifying the regularity of the diameter metric.   
	\item We applied this new indicator to low-dimensional toy models, allowing to clearly identify separatrices and chaotic seas stemming from stable-unstable manifolds. 
	\item We have presented numerical evidences on the concrete applications 
	and relevance 
	of the method to planetary problems, in the context of mean-motion resonances of the 2 and 3 planets problem.  
	We highlighted the benefits of the tool through dynamical maps, revealing secondary structures otherwise undetected using the MEM. 
\end{enumerate}

\appendix

\section{Application to a discrete case}\label{app:Discrete}
The framework presented applies also for nearly-integrable discrete systems. 
For illustrative purpose, following \cite{nGu17}, 
let us consider the $4$-dimensional mapping on $\mathbb{T}^{2} \times \mathbb{R}^{2}$ reading

\begin{align}\label{eq:DefMapping}
\left\{
\begin{aligned}
&x_{1}'=x_{1}+y_{1}, \\
&x_{2}'=x_{2}+y_{2}, \\
&y_{1}'=y_{1}-\frac{1}{2\pi}\big(
a\sin(2\pi x_{1})+c \sin(2\pi(x_{1}+x_{2}))
\big), \\
&y_{2}'=y_{2}-\frac{1}{2\pi}\big(
b\sin(2\pi x_{2})+c \sin(2\pi(x_{1}+x_{2}))
\big), 
\end{aligned}
\right.
\end{align}
where $a, b, c$ are real parameters. 
When $c=0$, the mapping is a product of two uncoupled  standard-maps. 
In the following, we consider $a=0.1$, $b=0.1$ and $c=0.07$, and we generate $500 \times 500$ initial conditions distributed in the $(y_{1},y_{2})$ action plane $[-0.25,0.65]^{2}$ (fixing $x_{1}=x_{2}=0$) iterated up to the final time $T=10^{3}$. The numerical setting  follows closely \cite{nGu17},
who dealt  with fast Lyapunov indicators \cite{cFr97}. 
In Fig.~\ref{figApp},  we show alternatively the results of the $D$ and  $\norm{\Delta D}$ analysis to provide a global representation of the phase space. Although the diameter reveal the low order resonant strips, it does not provide clear insights about the geography of lower order resonances, and the distribution of chaotic motions near the resonant crossings. This ``flatness'' in the map is reinflated by the $\norm{\Delta D}$ index, which reinvigorates minute details of the geography of low-order resonant structures.

\begin{figure}
	\centering
	\includegraphics[width=1\textwidth]{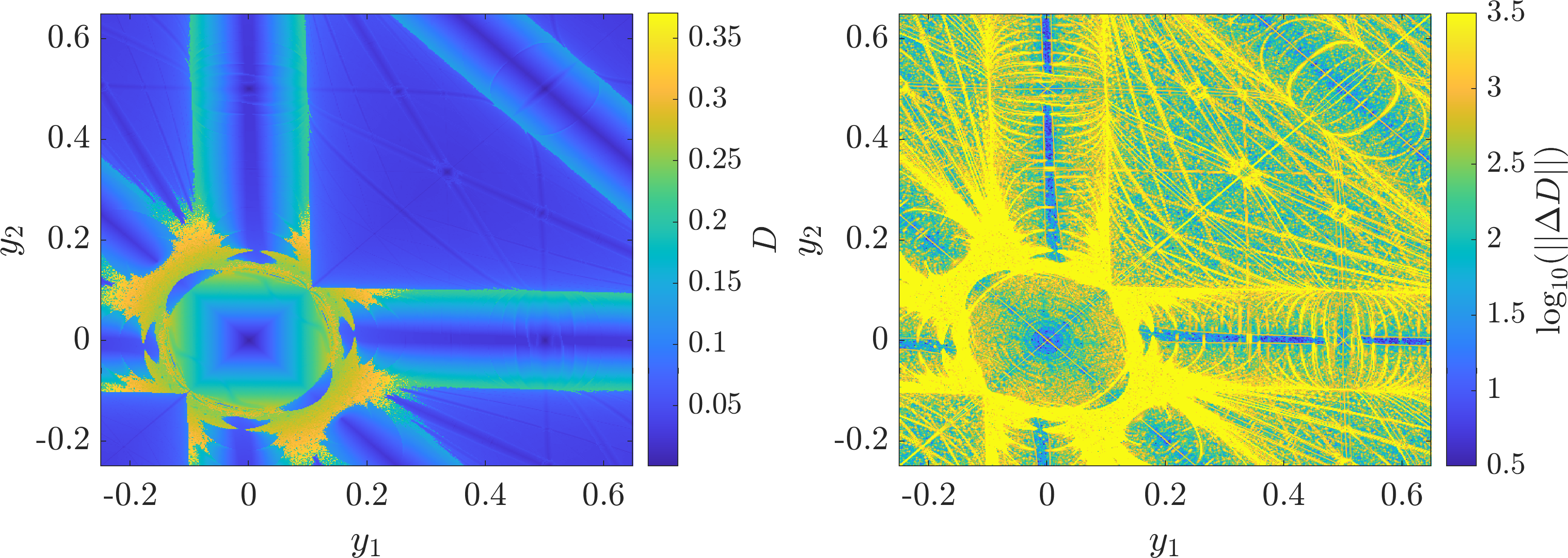}
	\caption{
		Dynamical maps associated to Eq.\,(\ref{eq:DefMapping}) with $a=0.1$, $b=0.1$, $c=0.07$.  Besides revealing the main strips, the $\norm{\Delta D}$ index (right-hand side) also highlights smaller resonant strips and chaos in the vicinity of their crossings, otherwise undetected with the diameter metric (left-hand side). 
	}
	\label{figApp}
\end{figure}

\section{Analytical properties of the diameter for the pendulum model}\label{app:AnalProof}
The non-differentiability of the diameter near the stable equilibrium and its discontinuity when crossing transversally the separatrix of the pendulum (as observed numerically in Fig.\,\ref{fig:fig1}) are proven analytically. 

\begin{prop}[Diameter in elliptic region.]\label{prop:V}
	Let $\mathcal{H}(p,q)=\frac{p^{2}}{2}+\frac{q^{2}}{2}$ be
	the Hamiltonian of the linear oscillator, $(p,q) \in D \subset \mathbb{R}^{2}$.
	Then we have
	\begin{equation}
	D(p_{0},q_{0})
	=
	2\sqrt{p_{0}^{2}+q_{0}^{2}}
	.
	\end{equation}
\end{prop}

\begin{proof}
	The system is $1$-\DOF and integrable. Following \cite{rPO21}, we parameterise orbits with their energy levels, thus accounting for an infinitely large time-window. 
	The flow generates circles around the origin with radii
	\begin{equation}
	r(p_{0},q_{0})=\sqrt{p_{0}^{2}+q_{0}^{2}}.
	\end{equation}
	The diameter thus reads 
	\begin{equation}
	D(p_{0},q_{0})=2r(p_{0},q_{0}).
	\end{equation}
	Along the line of initial condition $q_{0}=0$, one get
	\begin{equation}
	D(p_{0},0)=2\sqrt{p_{0}^{2}}=2\vert p_{0}\vert,
	\end{equation}
	and in particular $D$ is not differentiable at $p_{0}=0$.  
\end{proof}

\begin{prop}[Discontinuity when crossing the separatrix]
	Let $\mathcal{H}(p,q)=\frac{p^{2}}{2}-\cos q$ be the Hamiltonian of the pendulum, $(p,q) \in D \times [-\pi,\pi]$, $D \subset \mathbb{R}$. The diameter is discontinuous on the energy level labelling the separatrix.   
\end{prop}

\begin{proof}
	The librational domain corresponds to the range of energy $E \in [-1,1)$, the circulational domain to $E >1$, and the separatrix has energy $E=1$. Let $E_{0}$ denote the initial energy associated to $(p_{0},q_{0})$.  
	The diameter reads
	\begin{align}
	D(p_{0},q_{0})=
	\left\{
	\begin{aligned}
	& 2 \sqrt{2(E_{0}+1)}, \, E_{0} \in [-1,1), \\
	& \sqrt{2(E_{0}+1)} -  \sqrt{2(E_{0}-1)}, \, E_{0} > 1, 
	\end{aligned}  
	\right.  
	\end{align}
	from which follows the discontinuity announced at $E_{0}=1$. 
\end{proof}




%

\section*{Acknowledgments}
	J.\,D. is a postdoctoral researcher of the ``Fonds de la Recherche Scientifique'' - FNRS.
	C.C. acknowledges FNRS Grant No. F.4523.20 (DYNAMITE MIS-project). 
	It is our pleasure to acknowledge feedback and discussions with Ana Maria Mancho, Elisa Maria Alessi and Timoteo Carletti.

\bibliographystyle{plain}   
\bibliography{biblio}

\end{document}